\documentclass[11pt]{llncs}
\usepackage{amsmath}
\usepackage{amssymb}
\usepackage{amsfonts}
\usepackage{epic,gastex}
\usepackage{eucal}
\usepackage{array}
\usepackage{url}

\DeclareMathOperator{\dt}{.}

\DeclareMathOperator{\spn}{\mathrm{span}}
\DeclareMathOperator{\excl}{\mathrm{excl}}
\DeclareMathOperator{\dupl}{\mathrm{dupl}}

\newcommand{\sa}{synchronizing automata}
\newcommand{\san}{synchronizing automaton}
\newcommand{\cra}{completely reachable automata}
\newcommand{\cran}{completely reachable automaton}
\newcommand{\sw}{reset word}

\newcommand{\rl}{reset threshold}
\newcommand{\rt}{reset threshold}
\newcommand{\scc}{strongly connected component}
\newcommand{\scn}{strongly connected}

\DeclareSymbolFont{rsfscript}{OMS}{rsfs}{m}{n}
\DeclareSymbolFontAlphabet{\mathrsfs}{rsfscript}

\begin{document}
\title{Completely Reachable Automata\thanks{Supported by the Russian Foundation for Basic Research, grant no.\ 16-01-00795, the Ministry of Education and Science of the Russian Federation, project no.\ 1.1999.2014/K, and the Competitiveness Program of Ural Federal University. The paper was written during the second author's stay at Hunter College of the City University of New York as Ada Peluso Visiting Professor of Mathematics and Statistics with a generous support from the Ada Peluso Endowment.}}

\titlerunning{Completely Reachable Automata}

\author{E. A. Bondar \and M. V. Volkov}

\authorrunning{E. A. Bondar, M. V. Volkov}

\tocauthor{E. A. Bondar, M. V. Volkov (Ekaterinburg, Russia)}

\institute{Institute of Mathematics and Computer Science\\ Ural Federal University, Lenina 51, 620000 Ekaterinburg, Russia\\
\email{bondareug@gmail.com, mikhail.volkov@usu.ru}}

\maketitle

\begin{abstract}
We present a few results and several open problems
concerning complete deterministic finite automata
in which every non-empty subset of the state set
occurs as the image of the whole state set under
the action of a suitable input word.

\end{abstract}

\section{Background and overview}
\label{sec:intro}

We consider the most classical species of finite automata, namely, complete deterministic automata. Recall that a \emph{complete deterministic finite automaton} (DFA) is a triple
$\mathrsfs{A}=\langle Q,\Sigma,\delta\rangle$, where $Q$ and $\Sigma$ are finite sets called the \emph{state set} and the \emph{input alphabet} respectively, and $\delta\colon Q\times\Sigma\to Q$ is a totally defined map called the \emph{transition function}. Let $\Sigma^*$ stand for the collection of all finite words over the alphabet $\Sigma$, including the empty word. The function $\delta$ extends to a function $Q\times\Sigma^*\to Q$ (still denoted by $\delta$) in the following natural way: for every $q\in Q$ and $w\in\Sigma^*$, we set $\delta(q,w):=q$ if $w$ is empty and $\delta(q,w):=\delta(\delta(q,v),a)$ if $w=va$ for some word $v\in\Sigma^*$ and some letter $a\in\Sigma$. Thus, via $\delta$, every word $w\in\Sigma^*$ induces a transformation of the set $Q$.

Let $\mathcal{P}(Q)$ stand for the set of all non-empty subsets of the set $Q$. The function $\delta$ can be further extended to a function $\mathcal{P}(Q)\times\Sigma^*\to\mathcal{P}(Q)$ (again denoted by $\delta$) by letting $\delta(P,w):=\{\delta(q,w)\mid q\in P\}$ for every non-empty subset $P\subseteq Q$.
Thus, the triple $\mathcal{P}(\mathrsfs{A}):=\langle\mathcal{P}(Q),\Sigma,\delta\rangle$ is a DFA again; this DFA is referred to as the \emph{powerset automaton} of $\mathrsfs{A}$.

Whenever we deal with a fixed DFA, we simplify our notation by suppressing the sign of the transition function; this means that we may introduce the DFA as the pair $\langle Q,\Sigma\rangle$ rather than the triple $\langle Q,\Sigma,\delta\rangle$ and may write $q\dt w$ for $\delta(q,w)$ and $P\dt w$ for $\delta(P,w)$.

Given a DFA $\mathrsfs{A}=\langle Q,\Sigma\rangle$, we say that a non-empty subset $P\subseteq Q$ is \emph{reachable} in $\mathrsfs{A}$ if $P=Q\dt w$ for some word $w\in\Sigma^*$. A DFA is called \emph{completely reachable} if every non-empty subset of its state set is reachable.

Let us start with an example that served as a first spark which ignited our interest in \cra. A DFA $\mathrsfs{A}=\langle Q,\Sigma\rangle$ is called \emph{synchronizing} if it has a reachable singleton, that is, $Q\dt w$ is a singleton for some word $w\in\Sigma^*$. Any such word $w$ is said to be a \emph{reset word} for the DFA. The minimum length of reset words for $\mathrsfs{A}$ is called the \emph{\rl} of $\mathrsfs{A}$. In~1964 \v{C}ern\'{y}~\cite{Ce64} constructed for each $n>1$ a \san\ $\mathrsfs{C}_n$ with $n$ states, 2 input letters, and \rl\ $(n-1)^2$. Recall the definition of $\mathrsfs{C}_n$. If we denote the states of $\mathrsfs{C}_n$ by $1,2,\dots,n$ and the input letters by $a$ and $b$, the actions of the letters are as follows:
\[
i\dt a:=\begin{cases}
i &\text{if } i<n,\\
1 &\text{if } i=n;
\end{cases}
\qquad
i\dt b:=\begin{cases}
i+1 &\text{if } i<n,\\
1 &\text{if } i=n.
\end{cases}
\]
The automaton $\mathrsfs{C}_n$ is shown in Fig.\,\ref{fig:cerny-n}.
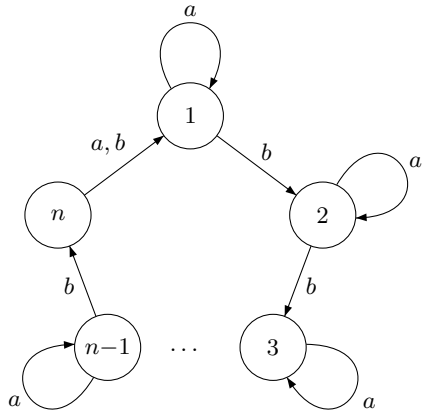
\begin{figure}[ht]
\begin{center}
\unitlength .55mm
\begin{picture}(72,85)(0,-83)
\gasset{Nw=16,Nh=16,Nmr=8}
\node(n0)(36.0,-16.0){1}
\node(n1)(4.0,-40.0){$n$} \node(n2)(68.0,-40.0){2}
\node(n3)(16.0,-72.0){$n{-}1$} \node(n4)(56.0,-72.0){3}
\drawedge[ELdist=2.0](n1,n0){$a,b$} \drawedge[ELdist=1.5](n2,n4){$b$}
\drawedge[ELdist=1.7](n0,n2){$b$}
\drawedge[ELdist=1.7](n3,n1){$b$}
\drawloop[ELdist=1.5,loopangle=30](n2){$a$}
\drawloop[ELdist=2.4,loopangle=-30](n4){$a$}
\drawloop[ELdist=1.5](n0){$a$}
\drawloop[ELdist=1.5,loopangle=210](n3){$a$}
\put(31,-73){$\dots$}
\end{picture}
\caption{The automaton $\mathrsfs{C}_n$}\label{fig:cerny-n}
\end{center}
\end{figure}

The automata in the \v{C}ern\'{y} series are well-known in the connection with the famous \v{C}ern\'{y} conjecture about the maximum \rl\ for \sa\ with $n$ states, see \cite{Vo08}. The automata $\mathrsfs{C}_n$ provide the lower bound $(n-1)^2$ for this maximum, and the conjecture claims that these automata represent the worst possible case since it has been conjectured that every \san\ with $n$ states can be reset by a word of length $(n-1)^2$. The automata $\mathrsfs{C}_n$ also have other interesting properties, including the one registered here:
\begin{example}
\label{examp:cerny}
Each automaton $\mathrsfs{C}_n$, $n>1$, is completely reachable.
\end{example}

The result of Example~\ref{examp:cerny} was first observed by Maslennikova~\cite[Proposition~2]{Maslennikova:2012}, see also \cite{Maslennikova:2014}, in the course of her study of the so-called reset complexity of regular ideal languages. Later, Don~\cite[Theorem~1]{Don:2015} found a sufficient condition for complete reachability that applies to the automata $\mathrsfs{C}_n$. In Section~\ref{sec:sufficient} we present another sufficient condition that both simplifies and generalizes Don's one. We provide an example showing that our condition is not necessary but we conjecture that it may be necessary for a stronger version of complete reachability.

In Section~\ref{sec:complexity} we discuss the problem of recognizing \cra. We show PSPACE-com\-plete\-ness of the following decision problem: given a DFA $\mathrsfs{A}=\langle Q,\Sigma\rangle$ and a subset $P\subseteq Q$, decide whether or not $P$ is reachable in $\mathrsfs{A}$. We also outline a polynomial algorithm that recognizes \cra\ with 2 input letters modulo the conjecture from Section~\ref{sec:sufficient}.

Given a DFA $\mathrsfs{A}=\langle Q,\Sigma\rangle$, its \emph{transition monoid} $M(\mathrsfs{A})$ is the monoid of all transformations of the set $Q$ induced by the words in $\Sigma^*$. By the \emph{syntactic complexity} of $\mathrsfs{A}$ we mean the size of $M(\mathrsfs{A})$. Clearly, the syntactic complexity of a \cran\ $\mathrsfs{A}$ with $n$ states cannot be less than $2^n-1$ since, for each non-empty subset $P$ of the state set, the transition monoid of $\mathrsfs{A}$ must contain a transformation whose image is $P$. In Section~\ref{sec:minimal} we address the question of the existence and classification of \emph{minimal} \cra, i.e., \cra\ with minimum possible syntactic complexity. This question has been recently investigated in the realm of transformation monoids by the first author~\cite{Bondar:2014,Bondar:2016}; here we translate her results into the language of automata theory and augment them by determining the input alphabet size of minimal \cra.

The present paper is in fact a work-in-progress report, and therefore, each of Sections~\ref{sec:sufficient}--\ref{sec:minimal} includes some open questions. Several additional open questions form Section~\ref{sec:final}; they mostly deal with synchronization properties of \cra.

We assume the reader's acquaintance with some basic concepts of graph theory, monoid theory, and computational complexity.

\section{A Sufficient Condition}
\label{sec:sufficient}

If $Q$ is a finite set, we denote by $T(Q)$ the \emph{full transformation monoid} on $Q$, i.e., the monoid consisting of all transformations $\varphi\colon Q\to Q$. For $\varphi\in T(Q)$, its \emph{defect} is defined as the size of the set $Q\setminus Q\varphi$. Observe that the defect of a product of transformations is greater than or equal to the defect of any of the factors and is equal to the defect of a factor whenever the other factors are permutations of $Q$. In particular, if a product of transformations has defect~1, then one of the factors must have defect~1.

Let $\mathrsfs{A}=\langle Q,\Sigma\rangle$ be a DFA. The defect of a word $w\in\Sigma^*$ with respect to $\mathrsfs{A}$ is the defect of transformation induced by $w$. Consider a word $w$ of defect~1. For such a word, the set $Q\setminus Q\dt w$ consists of a unique state, which is called the \emph{excluded state} for $w$ and is denoted by $\excl(w)$. Further, the set $Q\dt w$ contains a unique state $p$ such that $p=q_1\dt w=q_2\dt w$ for some $q_1\ne q_2$; this state $p$ is called the \emph{duplicate state} for $w$ and is denoted by $\dupl(w)$. Let $D_1(\mathrsfs{A})$ stand for the set of all words of defect~1 with respect to $\mathrsfs{A}$, and let $\Gamma_1(\mathrsfs{A})$ denote the directed graph having $Q$ as the vertex set and the set
\[
E_1:=\{(\excl(w),\dupl(w))\mid w\in D_1(\mathrsfs{A})\}
\]
as the edge set. Since we consider only directed graphs in this paper, we call them just graphs in the sequel. Recall that a graph is \emph{strongly connected} if for every pair of its vertices, there exists a directed path from the first vertex to the second.

\begin{theorem}
\label{thm:sufficient}
If a DFA $\mathrsfs{A}=\langle Q,\Sigma\rangle$ is such that the graph $\Gamma_1(\mathrsfs{A})$ is \scn, then $\mathrsfs{A}$ is completely reachable.
\end{theorem}

\begin{proof}
Take an arbitrary non-empty subset $P\subseteq Q$. We prove that $P$ is reachable in $\mathrsfs{A}$ by induction on $k:=|Q\setminus P|$. If $k=0$,
then $P=Q$ and nothing is to prove as $Q$ is reachable via the empty word. Now let $k>0$ so that $P$ is a proper subset of $Q$. Since the graph $\Gamma_1(\mathrsfs{A})$ is \scn, there exists an edge $(q,p)\in E_1$ that connects $Q\setminus P$ and $P$ in the sense that $q\in Q\setminus P$ while $p\in P$. By the definition of $E_1$, there exists a word $w$ of defect~1 with respect to $\mathrsfs{A}$ for which $q$ is the excluded state and $p$ is the duplicate state. By the definition of the duplicate state, $p=q_1\dt w=q_2\dt w$ for some $q_1\ne q_2$, and since the excluded state $q$ for $w$ does not belong to $P$, for each state $r\in P\setminus\{p\}$, there exists a unique state $r'\in Q$ such that $r=r'\dt w$. Now letting $R:=\{q_1,q_2\}\cup\bigl\{r'\mid r\in P\setminus\{p\}\bigr\}$, we conclude that $P=R\dt w$ and $|R|=|P|+1$. Then $|Q\setminus R|=k-1$, and the induction assumption applies to the subset $R$ whence $R=Q\dt v$ for some word $v\in\Sigma^*$. Then $P=Q\dt vw$ so that $P$ is reachable as required.
\end{proof}

Don~\cite{Don:2015} has formulated a sufficient condition for complete reachability in the terms of what he called a state map. Consider a DFA $\mathrsfs{A}=\langle Q,\Sigma\rangle$ with $n$ states in which every subset of size $n-1$ is reachable. Let $W$ be a set of $n$ words of defect~1 with respect to $\mathrsfs{A}$ such that for every subset $P\subset Q$ with $|P|=n-1$ there is a unique word $w\in W$ with $P=Q\dt w$. (Such a set is termed a 1-\emph{contracting collection} in~\cite{Don:2015}). The \emph{state map} $\sigma_W\colon Q\to Q$ induced by $W$ is defined by
\[
q\sigma_W:=\dupl(w)\ \text{ for $w\in W$ such that }\ q=\excl(w).
\]
The following is one of the main results in~\cite{Don:2015}:
\begin{theorem}
\label{thm:don}
A DFA $\mathrsfs{A}$ is completely reachable if it admits a $1$-con\-tract\-ing collection such that the induced state map is a cyclic permutation of the state set of $\mathrsfs{A}$.
\end{theorem}

Even though Theorem~\ref{thm:don} is stated in different terms, it is easily seen to constitute a special case of Theorem~\ref{thm:sufficient}. Indeed, if $W$ is a $1$-contracting collection and $\sigma_W$ is the corresponding state map, then each pair $(q,q\sigma_W)$ can be treated as an edge in $E_1$. Therefore, if $\sigma_W$ is a cyclic permutation of $Q$, then the set of edges $\{(q,q\sigma_W)\mid q\in Q\}$ forms a directed Hamiltonian cycle in the graph $\Gamma_1(\mathrsfs{A})$ whence the latter is \scn.

We believe that Theorem~\ref{thm:sufficient} may have strongly wider application range than Theorem~\ref{thm:don} even though at the moment we do not have any example confirming this conjecture. If the conditions of the two theorems were equivalent, every \scn\ graph of the form $\Gamma_1(\mathrsfs{A})$ would possess a directed Hamiltonian cycle, and this does not seem to be likely.

Now we demonstrate that the condition of Theorem~\ref{thm:sufficient} is not necessary.
\begin{example}
\label{examp:not necessary}
Consider the DFA $\mathrsfs{E}_3$ with the state set $\{1,2,3\}$ and the input letters $a_{[1]},a_{[2]},a_{[3]},a_{[1,2]}$ that act as follows:
\begin{eqnarray*}
i\dt a_{[1]}:=\begin{cases}
2 &\text{if } i=1,2,\\
3 &\text{if } i=3;
\end{cases}
&\qquad&
i\dt a_{[2]}:=\begin{cases}
1 &\text{if } i=1,2,\\
3 &\text{if } i=3;
\end{cases}
\\
i\dt a_{[3]}:=\begin{cases}
1 &\text{if } i=1,2,\\
2 &\text{if } i=3;
\end{cases}
&\qquad&
i\dt a_{[1,2]}:=3\ \text{ for all }\ i=1,2,3.
\end{eqnarray*}
The automaton $\mathrsfs{E}_3$ is shown in Fig.\,\ref{fig:e3} on the left.
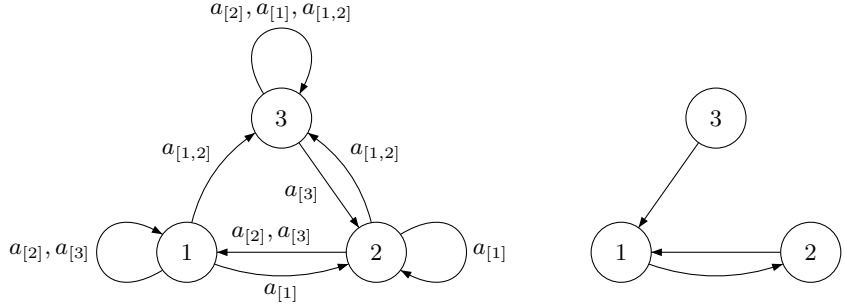
\begin{figure}[ht]
\begin{center}
\unitlength=1.05mm
\begin{picture}(80,32)(-7.5,-5)
\node(A1)(0,0){1} \node(B1)(24,0){2} \node(C1)(12,17){3}
\drawedge[ELside=r,curvedepth=-3](A1,B1){$a_{[1]}$}
\drawedge[ELside=r,ELpos=55](B1,A1){$a_{[2]},a_{[3]}$}
\drawedge[ELside=r,ELpos=45](C1,B1){$a_{[3]}$}
\drawedge[ELside=r,curvedepth=-3](B1,C1){$a_{[1,2]}$}
\drawedge[curvedepth=3](A1,C1){$a_{[1,2]}$}
\drawloop[ELdist=.8,loopangle=0](B1){$a_{[1]}$}
\drawloop[ELdist=.8,loopangle=180](A1){$a_{[2]},a_{[3]}$}
\drawloop[ELdist=.8](C1){$a_{[2]},a_{[1]},a_{[1,2]}$}
\node(A2)(55,0){1} \node(B2)(79,0){2} \node(C2)(67,17){3}
\drawedge[ELside=r,curvedepth=-3](A2,B2){}
\drawedge[ELside=r,ELpos=55](B2,A2){}
\drawedge[ELside=r,ELpos=45](C2,A2){}
\end{picture}
\caption{The automaton $\mathrsfs{E}_3$ and the graph $\Gamma_1(\mathrsfs{E}_3)$}\label{fig:e3}
\end{center}
\end{figure}
The graph $\Gamma_1(\mathrsfs{E}_3)$ is shown in Fig.\,\ref{fig:e3} on the right; it is not \scn. However, it can be checked by a straightforward computation that the automaton $\mathrsfs{E}_3$ is completely reachable.
\end{example}

%We mention in passing that the automaton $\mathrsfs{E}_3$ can be used to answer in the negative a question asked in \cite{Don:2015}, namely, Question~1 in page~15.

The reason of why the converse of Theorem~\ref{thm:sufficient} fails becomes obvious if one analyzes the above proof. In fact, we have proved more than we have formulated, namely, our proof shows that if a DFA $\mathrsfs{A}$ is such that the graph $\Gamma_1(\mathrsfs{A})$ is \scn, then \emph{every proper non-empty subset of the state set of $\mathrsfs{A}$ is reachable via a product of words of defect}~1. Of course, this stronger property has no reason to hold in an arbitrary \cran. For instance, in the automaton $\mathrsfs{E}_3$ of Example~\ref{examp:not necessary} the singleton $\{3\}$ is not an image of any product of words of defect~1. On the other hand, for the stronger property italicized above, the condition of Theorem~\ref{thm:sufficient} may be not only sufficient but also necessary. We formulate this guess as a conjecture.

\begin{conjecture}
\label{conj:necessary}
If for every proper non-empty subset $P$ of the state set of a DFA $\mathrsfs{A}$ there is a product $w$ of words of defect~$1$ with respect to $\mathrsfs{A}$ such that $P=Q\dt w$, the graph $\Gamma_1(\mathrsfs{A})$ is \scn.
\end{conjecture}

One can formulate further sufficient conditions for complete reachability in terms of strong connectivity of certain \emph{hypergraphs} related to words of defect 2.

\section{Complexity of Deciding Reachability}
\label{sec:complexity}

Given a DFA, one can easily decide whether or not it is completely reachable considering its powerset automaton: a DFA $\mathrsfs{A}=\langle Q,\Sigma\rangle$ is completely reachable if and only if $Q$ is connected with every its non-empty subset by a directed path in the powerset automaton $\mathcal{P}(\mathrsfs{A})$, and the latter property can be recognized by breadth-first search on $\mathcal{P}(\mathrsfs{A})$ starting at $Q$. This algorithm is however exponential with respect to the size of $\mathrsfs{A}$, and it is natural to ask   whether or not complete reachability can be decided in polynomial time. First, consider the following decision problem:

\smallskip

\hangindent=\parindent \noindent\textsc{Reachable Subset}: \emph{Given a DFA $\mathrsfs{A}=\langle Q,\Sigma,\delta\rangle$ and a non-empty subset $P\subseteq Q$,
is it true that $P$ is reachable in $\mathrsfs{A}$?}

\begin{theorem}
\label{thm:complexity}
The problem \textsc{Reachable Subset} is \textup{PSPACE}-complete.
\end{theorem}

\begin{proof}
The fact that \textsc{Reachable Subset} is in the class PSPACE is easy and known, see, e.g., \cite[Lemma 6, item 1]{Brandl&Simon:2015}.

To prove PSPACE-hardness of \textsc{Reachable Subset}, we reduce to it in logarithmic space the well-known PSPACE-complete problem \textsc{FAI} (\textsc{Finite Automata Intersection}, see \cite{Kozen:1977}). Recall that an instance of \textsc{FAI} consists of $k$ DFAs $\mathrsfs{A}_j=\langle Q_j,\Sigma,\delta_j\rangle$, $j=1,\dots,k$, with disjoint state sets and a common input alphabet. In each DFA $\mathrsfs{A}_j$ an \emph{initial state} $s_j\in Q_j$ and a \emph{final state} $t_j\in Q_j$ are specified; a word $w\in\Sigma^*$ is said to be \emph{accepted} by $\mathrsfs{A}_j$ if $\delta_j(s_j,w)=t_j$. The question of \textsc{FAI} asks whether or not there exists a word $w\in\Sigma^*$ which is simultaneously accepted by all automata $\mathrsfs{A}_1,\dots,\mathrsfs{A}_k$.

Now, given an instance of \textsc{FAI} as above, we construct the following instance $(\mathrsfs{A},P)$ of \textsc{Reachable Subset}. The state set of the DFA $\mathrsfs{A}$ is $Q:=\bigcup_{j=1}^kQ_j$; the input alphabet of $\mathrsfs{A}$ is $\Sigma$ with one extra letter $\rho$ added. The transition function $\delta\colon Q\times(\Sigma\cup\{\rho\})\to Q$ is defined by the rule
\begin{equation}
\label{eq:rule}
\delta(q,a):=\begin{cases}
\delta_j(q,a)&\text{if $a\in\Sigma$ and $q\in Q_j$},\\
s_j&\text{if $a=\rho$ and $q\in Q_j$}.
\end{cases}
\end{equation}
Expressing this rule less formally, it says that, given a state $q\in Q$, one first should find the index $j\in\{1,\dots,k\}$ such that $q$ belongs to $Q_j$; then every letter $a\in\Sigma$ acts on $q$ in the same way as it does in the automaton $\mathrsfs{A}_j$ while the added letter $\rho$ sends $q$ to the initial state $s_j$ of $\mathrsfs{A}_j$ (so $\rho$ artificially `initializes' each $\mathrsfs{A}_j$). Observe that each set $Q_j$ is closed under the action of each letter in $\Sigma\cup\{\rho\}$. Finally, we set $P:=\{t_1,\dots,t_k\}$, that is, $P$ consists of the final states of $\mathrsfs{A}_1,\dots,\mathrsfs{A}_k$.

We claim that the subset $P$ is reachable in $\mathrsfs{A}$ if and only if there exists a word $w\in\Sigma^*$ which is simultaneously accepted by all automata $\mathrsfs{A}_1,\dots,\mathrsfs{A}_k$. Indeed, if such a word $w$ exists, then $\delta(Q,\rho w)=P$ since we have $\delta(Q,\rho)=\{s_1,\dots,s_k\}$ by \eqref{eq:rule} and
$\delta(s_j,w)=\delta_j(s_j,w)=t_j$ for each $j=1,\dots,k$ by the choice of $w$. Conversely, suppose that $P$ is reachable in $\mathrsfs{A}$, that is, $\delta(Q,u)=P$ for some word $u\in(\Sigma\cup\{\rho\})^*$. Then we must have $\delta_j(Q_j,u)=\{t_j\}$ for each $j=1,\dots,k$. If the word $u$ has no occurrence of the letter $\rho$, then $u\in\Sigma^*$ and
$\delta_j(s_j,u)=\{t_j\}$ for each $j=1,\dots,k$ so that $u$ is simultaneously accepted by all automata $\mathrsfs{A}_1,\dots,\mathrsfs{A}_k$. Otherwise we fix the rightmost occurrence of $\rho$ in $u$ and denote by $w$ the suffix of $u$ following this occurrence so that $w\in\Sigma^*$ and $u=v\rho w$ for some $v\in(\Sigma\cup\{\rho\})^*$. Then
$\delta_j(Q_j,v\rho)=\{s_j\}$ and $\delta_j(s,w)=\delta(Q_j,v\rho w)=\{t_j\}$ for each $j=1,\dots,k$. We conclude that $w$ is simultaneously accepted by all automata $\mathrsfs{A}_1,\dots,\mathrsfs{A}_k$. This completes the proof of our claim and establishes the reduction which obviously can be implemented in logarithmic space.
\end{proof}

The reduction used in the above proof is an adaptation of a slightly more involved log-space reduction used by Brandl and Simon~\cite[Section~3]{Brandl&Simon:2015} to show PSPACE-hardness of a natural problem about transformation monoids presented by a bunch of generating transformations. Using a trick from Martyugin's paper~\cite{Martyugin:2014},
one can modify this reduction to show that \textsc{Reachable Subset} remains PSPACE-complete even if restricted to automata with only 2 input letters. 

In connection with Theorem~\ref{thm:complexity}, an interesting result by Goral\v{c}\'{\i}k and Koubek~\cite[Theorem~1]{Goralcik&Koubek:1995} is worth being mentioned. If stated in the language adopted in the present paper, their result says that, given a DFA $\mathrsfs{A}=\langle Q,\Sigma\rangle$ with $|Q|=n$, $|\Sigma|=m$ and a subset $P\subseteq Q$ with $|P|=k$, one can decide in $O\big((k+1)n^{k+1}m\big)$ time whether or not there exists a word $w\in\Sigma^*$ such that $P=Q\dt w=P\dt w$. (The difference from our definition of reachability is that here one looks for a word not only having the subset $P$ as its image but also acting on $P$ as a permutation.) Thus, if the size of the target set $P$ is treated as a parameter, the algorithm from~\cite{Goralcik&Koubek:1995} becomes polynomial. One can ask if a similar result holds for the parameterized version of \textsc{Reachable Subset} formulated as follows:

\smallskip

\hangindent=\parindent \noindent\textsc{Reachable Subset}$_k$: \emph{Given a DFA $\mathrsfs{A}=\langle Q,\Sigma,\delta\rangle$ and a non-empty subset $P\subseteq Q$ of size $k$,
is it true that $P$ is reachable in $\mathrsfs{A}$?}

\smallskip

For $k=1$, the cited result by Goral\v{c}\'{\i}k and Koubek applies since, for $P$ being a singleton, any word $w\in\Sigma^*$ such that $P=Q\dt w$ automatically satisfies the additional condition $P\dt w=P$. For $k>1$, the question about the exact complexity of \textsc{Reachable Subset}$_k$ is open. The reduction from the proof of  Theorem~\ref{thm:complexity} cannot help here because the size $k$ of the subset $P$ in this reduction is equal to the number of DFAs in the instance of \textsc{FAI} from which we depart, and for each fixed $k$, there is a polynomial algorithm that decides on all instances of \textsc{FAI} with $k$ automata. Pribavkina and Rodaro~\cite[Sections~7 and~8]{Pribavkina&Rodaro:2011} used \textsc{Reachable Subset}$_2$ as an intermediate problem in their study of so-called \emph{finitely generated} \sa; from their results it follows that \textsc{Reachable Subset}$_2$ is co-NP-hard even if restricted to automata with only 2 input letters.  

Now we return to the question of whether or not complete reachability can be decided in polynomial time. It should be noted that Theorem~\ref{thm:complexity} does not imply any hardness conclusion here: while checking reachability of individual subsets is PSPACE-complete, checking reachability of all non-empty subsets may still be polynomial even though the latter problem consists of exponentially many individual problems! One can illustrate this phenomenon of `simplification due to collectivization' with the following example.
If is known~\cite{Kozen:1977} that the following \emph{membership problem} for transition monoids of DFAs is PSPACE-complete: given a DFA $\mathrsfs{A}=\langle Q,\Sigma\rangle$ and a transformation $\varphi\colon Q\to Q$, does $\varphi$ belongs to the transition monoid $M(\mathrsfs{A})$, i.e., is there a word $w\in\Sigma^*$ such that $q\varphi=q\dt w$ for all $q\in Q$? On the other hand, one can decide in polynomial time whether or not \emph{every} transformation of the state set belongs to the transition monoid of a given DFA. Indeed, given a DFA $\mathrsfs{A}=\langle Q,\Sigma\rangle$, we partition the alphabet $\Sigma$ as $\Sigma=\Pi\cup\Delta$, where $\Pi$ consists of all letters that act on $Q$ as permutations and $\Delta$ contains all letters with non-zero defect. First we inspect $\Delta$: if no letter in $\Delta$ has defect~1, then it is clear that the monoid $M(\mathrsfs{A})$ contains no transformation of defect~1 (see the observation registered at the beginning of Section~\ref{sec:sufficient}). Further, we  invoke twice the polynomial algorithm by Furst, Hopcroft and Luks~\cite{Furst&Hopcroft&Luks:1980} for the membership problem in permutation groups: we fix a cyclic permutation and a transposition of $Q$ and check if they belong to the permutation group on $Q$ generated by the permutations induced by the letters in $\Pi$. If the answers to all these queries are affirmative, then  $M(\mathrsfs{A})$ contains a cyclic permutation, a transposition, and a transformation of defect 1, and it is well-known that any such trio of transformations generates the full transformation monoid $T(Q)$, see, e.g., \cite[Theorem~3.1.3]{Ganyushkin&Mazorchuk:2009}.

Thus, the complexity of deciding complete reachability for a given DFA remains unknown so far. We expect this problem to be computationally hard for automata over unrestricted alphabets while for automata with a fixed number of letters a polynomial algorithm may exist. For instance, if Conjecture~\ref{conj:necessary} holds true, there exists a polynomial algorithm that recognizes \cra\ among DFAs with 2 input letters. Indeed, let $\mathrsfs{A}=\langle Q,\{a,b\}\rangle$ be a DFA with $n$ states, $n>1$. Every subset of the form $Q\dt w$, where $w$ is a non-empty word over $\{a,b\}$, is contained in either $Q\dt a$ or $Q\dt b$. At least one of the letters must have defect~1 since no subset of size $n-1$ is reachable otherwise, and if the other letter has defect greater than 1, only one subset of size $n-1$ is reachable. Hence, if $\mathrsfs{A}$ is a \cran, one of its letters  has defect~1 while the other has defect at most 1. Therefore for each proper reachable subset $P\subset Q$, there is a product $w$ of words of defect~$1$ with respect to $\mathrsfs{A}$ such that $P=Q\dt w$. In view of Theorem~\ref{thm:sufficient}, if Conjecture~\ref{conj:necessary} holds true, then complete reachability of $\mathrsfs{A}$ is equivalent to strong connectivity of the graph $\Gamma_1(\mathrsfs{A})$. It remains to show that for automata with 2 input letters, the latter condition can be verified in polynomial time.

Once the graph $\Gamma_1(\mathrsfs{A})$ is constructed, checking its strong connectivity in polynomial time makes no difficulty. However, it is far from being obvious that $\Gamma_1(\mathrsfs{A})$, even though it definitely has polynomial size, can always be constructed in polynomial time. Indeed, by the definition, the edges of $\Gamma_1(\mathrsfs{A})$ arise from transformations of defect~1 in the transition monoid of $\mathrsfs{A}$, and for an automaton with $n$ states, the number of transformations of defect~1 in $M(\mathrsfs{A})$ may reach $n!\binom{n}2$. Our algorithm depends on some peculiarities of automata with 2~input letters. It incrementally appends edges to a spanning subgraph of $\Gamma_1(\mathrsfs{A})$ in a way such that one can reach a conclusion about strong connectivity of $\Gamma_1(\mathrsfs{A})$ by examining only polynomially many transformations of defect~1. In the following brief and rather informal description of the algorithm, we use the notation introduced in Section~\ref{sec:sufficient} in the course of defining the graph $\Gamma_1(\mathrsfs{A})$.

Thus, again, let $\mathrsfs{A}=\langle Q,\{a,b\}\rangle$ be a DFA with $n$ states, $n>1$. For certainty, let $a$ stand for the letter of defect~1. If $b$ also has defect~1, then at most two subsets of size $n-1$ are reachable (namely, $Q\dt a$ and $Q\dt b$), and $\mathrsfs{A}$ can only be completely reachable provided that $n=2$.
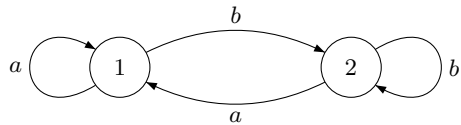
\begin{figure}[ht]
\begin{center}
  \unitlength=3.5pt
    \begin{picture}(25,12)(0,-7)
    \gasset{curvedepth=4}
    \node(A1)(0,0){$1$}
    \drawloop[loopangle=180](A1){$a$}
    \node(A2)(25,0){$2$}
    \drawloop[loopangle=0](A2){$b$}
    \drawedge(A1,A2){$b$}
    \drawedge(A2,A1){$a$}
    \end{picture}
\caption{Filp-flop}\label{fig:flip-flop}
\end{center}
\end{figure}
The automaton $\mathrsfs{A}$ is then nothing but the classical flip-flop, see Fig.~\ref{fig:flip-flop}. Beyond this trivial case, $b$ must be a permutation of $Q$ whence $b^n$ acts on $Q$ as the identity transformation. Then the set $\{\excl(w)\mid w\in D_1(\mathrsfs{A})\}$ of the states at which edges of $\Gamma_1(\mathrsfs{A})$ may originate is easily seen to coincide with the set $\{\excl(a),\excl(ab),\dots,\excl(ab^{n-1})\}$. For $\Gamma_1(\mathrsfs{A})$ to be \scn, it is necessary that every vertex is an origin of an edge whence the latter set must be equal to $Q$. Taking into account that $\excl(ab^k)=\excl(a)\dt b^k$ for each $k=1,\dots,n-1$, we conclude that $b$ must be a cyclic permutation of $Q$. It is easy to show that $\excl(w)\dt b=\excl(wb)$ and $\dupl(w)\dt b=\dupl(wb)$ for every word $w$ of defect~1, and therefore, $b$ acts as a permutation on the edge set $E_1$ of $\Gamma_1(\mathrsfs{A})$.

The set $E_1$ contains the edges
\begin{equation}
\label{eq:edges}
(\excl(a),\dupl(a)),\dots,(\excl(ab^{n-1}),\dupl(ab^{n-1})).
\end{equation}
Since $\dupl(ab^k)=\dupl(a)\dt b^k$ for each $k=1,\dots,n-1$, the edges in~\eqref{eq:edges} are the `translates' of the edge $(\excl(a),\dupl(a))$. Any two edges in \eqref{eq:edges} start at different vertices and end at different vertices, whence for some $d$ such that $d<n$ and $d$ divides $n$, the edges in \eqref{eq:edges} form $d$ directed cycles, each of size $\frac{n}d$. If $d=1$, we can already conclude that the graph $\Gamma_1(\mathrsfs{A})$ is \scn. If $d>1$, denote the cycles by $C_1,\dots,C_d$ and consider the words $a^2, aba, \dots, ab^{n-1}a$. It can be easily shown that exactly two of them have defect~1; let us denote these two words by $w_1$ and $w_2$. Since $w_1$ and $w_2$ end with $a$, we have $Q\dt w_1=Q\dt w_2=Q\dt a$ whence $\excl(w_1)=\excl(w_2)=\excl(a)$. Thus, the edges $(\excl(w_1),\dupl(w_1))$ and $(\excl(w_2),\dupl(w_2))$ start at the vertex $\excl(a)$ which can be assumed to belong to the cycle $C_1$. If also the ends of these edges lie in $C_1$, one can show that no further edge in $E_1$ can connect $C_1$ with another cycle whence $C_1$ forms a \scc\ of $\Gamma_1(\mathrsfs{A})$. We then conclude that $\Gamma_1(\mathrsfs{A})$ is not \scn.

Now suppose that the edge $(\excl(w_i),\dupl(w_i))$ where $i=1$ or $i=2$ connects the vertex $\excl(a)$ with a vertex from the cycle $C_j$ where $j>1$. Then we append the edge and all its translates $(\excl(w_ib^k),\dupl(w_ib^k))$, $k=1,\dots,n-1$, to $C_1,\dots,C_d$; in the case where both $(\excl(w_1),\dupl(w_1))$ and $(\excl(w_2),\dupl(w_2))$ leave $C_1$, we append both these edges and all their translates. After that, we get larger \scn\ subgraphs $D_1,\dots, D_\ell$ isomorphic to each other, where $\ell<d$ and $\ell$ divides $d$. If $\ell=1$, then the graph $\Gamma_1(\mathrsfs{A})$ is \scn. If $\ell>1$, we iterate by considering the words $w_ia, w_iba, \dots, w_ib^{n-1}a$. Eventually, either we reach a \scn\ spanning subgraph of $\Gamma_1(\mathrsfs{A})$, and then the graph $\Gamma_1(\mathrsfs{A})$ is \scn\ as well, or on some step the process gets stacked, which means that $\Gamma_1(\mathrsfs{A})$ has a proper \scc, and therefore, is not \scn.

The described process branches, and in the worst case the number of words of defect~1 to be analyzed doubles at each step. On the other hand, since the steps are indexed by a chain of divisors of $n$, the number of steps does not exceed $\log_2n+1$. Thus, executing the algorithm, we have to analyze at most
\[
1+2+4+\dots +2^{\lceil\log_2n\rceil+1}=O(n)
\]
words of maximum length $O(n\log_2n)$, and therefore, the algorithm can be implemented in polynomial time.

\smallskip

We illustrate the above algorithm by running it on the DFA $\mathrsfs{E}_6$ with the state set $\{1,2,3,4,5,6\}$ and the input letters $a,b$ that act as follows:
\[
i\dt a:=\begin{cases}
i+1 &\text{if } i=1,2,\\
9-i &\text{if } i=3,6,\\
i &\text{if } i=4,5;
\end{cases}
\qquad
i\dt b:=\begin{cases}
i+1 &\text{if } i<6,\\
1 &\text{if } i=6.
\end{cases}
\]
The automaton $\mathrsfs{E}_6$ is shown in Fig.\,\ref{fig:e6}.
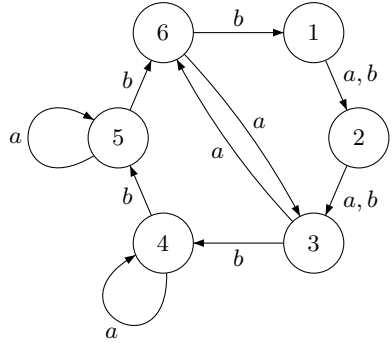
\begin{figure}[ht]
\begin{center}
\unitlength .5mm
\begin{picture}(72,72)(0,-86)
\gasset{Nw=16,Nh=16,Nmr=8} \node(n1)(56.0,-16.0){1} \node(n2)(68.0,-44.0){2}
\node(n3)(56.0,-72.0){3} \node(n4)(16.0,-72.0){4}
\node(n5)(4.0,-44.0){5} \node(n6)(16.0,-16.0){6}
\drawedge[ELdist=1.7](n1,n2){$a,b$} \drawedge[ELdist=1.7](n2,n3){$a,b$}
\drawedge[ELdist=1.7](n3,n4){$b$}
\drawedge[ELdist=1.7](n4,n5){$b$} \drawedge[ELdist=1.7](n5,n6){$b$}
\drawedge[ELdist=1.7](n6,n1){$b$}
\drawedge[curvedepth=3,ELdist=1.7](n6,n3){$a$}
\drawedge[curvedepth=3,ELdist=1.7](n3,n6){$a$}
\drawloop[ELdist=1.5,loopangle=240](n4){$a$}
\drawloop[ELdist=1.5,loopangle=180](n5){$a$}
\end{picture}
\end{center}
\caption{The automaton $\mathrsfs{E}_6$}\label{fig:e6}
\end{figure}

\begin{example}
\label{examp:e6}
The automaton $\mathrsfs{E}_6$ is completely reachable.
\end{example}

We verify the claim of Example~\ref{examp:e6} by constructing a \scn\ spanning subgraph of the graph $\Gamma_1(\mathrsfs{E}_6)$. As $\excl(a)=1$ and $\dupl(a)=3$, we see that the edge set $E_1$ of $\Gamma_1(\mathrsfs{E}_6)$ contains the edge $1\to 3$, and hence, also its translates $2\to 4$, $3\to 5$, $4\to 6$, $5\to 1$, and $6\to 2$. Altogether, we have 6 edges of the form~\eqref{eq:edges}; they make 2 disjoint cycles $1\to 3\to 5\to 1$ and $2\to 4\to 6\to 2$, and hence, after the first step of the algorithm we cannot yet exhibit a
\scn\ spanning subgraph in $\Gamma_1(\mathrsfs{E}_6)$.

We then proceed by inspecting the words $a^2, aba, ab^2a, ab^3a, ab^4a, ab^5a$. The actions of these words are shown in the following table:
\begin{center}
\begin{tabular}{c@{\rule{5pt}{0pt}}|@{\rule{5pt}{0pt}}c@{\rule{5pt}{0pt}}c@{\rule{5pt}{0pt}}c@{\rule{5pt}{0pt}}c@{\rule{5pt}{0pt}}c@{\rule{5pt}{0pt}}c}
        & 1 & 2 & 3 & 4 & 5 & 6\\
\hline
$a^2$\rule{0pt}{13pt}   & 3 & 6 & 3 & 4 & 5 & 6\\
$aba$\rule{0pt}{13pt}   & 6 & 4 & 2 & 5 & 3 & 4\\
$ab^2a$\rule{0pt}{13pt} & 4 & 5 & 3 & 3 & 2 & 5\\
$ab^3a$\rule{0pt}{13pt} & 5 & 3 & 6 & 2 & 3 & 3\\
$ab^4a$\rule{0pt}{13pt} & 3 & 2 & 4 & 3 & 6 & 2\\
$ab^5a$\rule{0pt}{13pt} & 2 & 3 & 5 & 6 & 4 & 3
\end{tabular}
\end{center}
We see that exactly two words, namely, $aba$ and $ab^5a$, are of defect~1. Since $\dupl(ab^5a)=3$, we discard the word $ab^5a$ because it induces the same edge $1\to 3$ as $a$. In contrast, $\dupl(aba)=4$ whence the word $aba$ adds the edge $1\to 4$ to $E_1$. If we add also the translates $2\to 5$, $3\to 6$, $4\to 1$, $5\to 2$, and $6\to 3$ of this edge, we obtain the \scn\ spanning subgraph in $\Gamma_1(\mathrsfs{E}_6)$ shown in Fig~\ref{fig:graph for e6}.
\begin{figure}[ht]
\begin{center}
\unitlength=.85mm
\begin{picture}(60,45)(0,5)\nullfont
\node(n1)(10,10){1} \node(n2)(10,30){3} \node(n3)(10,50){5}
\node(n4)(50,10){4} \node(n5)(50,30){6} \node(n6)(50,50){2}
\drawedge(n1,n2){} \drawedge(n2,n3){}
\drawedge[curvedepth=-10](n3,n1){}
\drawedge(n4,n5){} \drawedge(n5,n6){}
\drawedge[curvedepth=10](n6,n4){}
\drawedge[curvedepth=2](n1,n4){}
\drawedge[curvedepth=2](n4,n1){}
\drawedge[curvedepth=2](n2,n5){}
\drawedge[curvedepth=2](n5,n2){}
\drawedge[curvedepth=2](n3,n6){}
\drawedge[curvedepth=2](n6,n3){}
\end{picture}
\caption{A  \scn\ spanning subgraph in the graph $\Gamma_1(\mathrsfs{E}_6)$}\label{fig:graph for e6}
\end{center}
\end{figure}
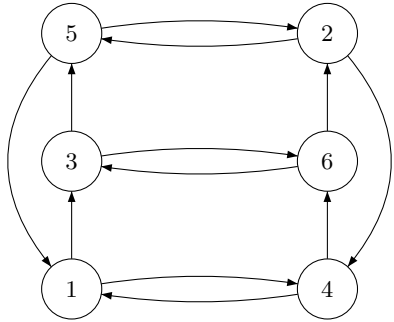
Thus, the graph $\Gamma_1(\mathrsfs{E}_6)$ is \scn\ and the claim of Example~\ref{examp:e6} follows from Theorem~\ref{thm:sufficient}.

\smallskip

Don~\cite[Proposition~2]{Don:2015} has found a sufficient condition for complete reachability of a DFA $\mathrsfs{A}$ that has a letter ($a$, say) of defect 1 and a letter ($b$, say) that act as a cyclic permutation of the state set. Namely, let $d$ be the least integers such that $\excl(a)\dt b^d=\dupl(a)$. If $d$ and the number of states in $\mathrsfs{A}$ are coprime, then $\mathrsfs{A}$ is completely reachable. Our Example~\ref{examp:e6} demonstrates that this condition is not necessary: the automaton $\mathrsfs{E}_6$ is completely reachable while the parameter $d$ for this DFA equals 2 and divides the number of states. In fact, it is easy to see that Don's condition precisely characterizes DFAs $\mathrsfs{A}$ for which our algorithm produces a  \scn\ spanning subgraph in the graph $\Gamma_1(\mathrsfs{A})$ after the first step.

We conclude this section with another example that shows the behaviour of our algorithm in the situation where the DFA under investigation is not completely reachable. Consider 
the DFA $\mathrsfs{E}'_6$ with the state set $\{1,2,3,4,5,6\}$ and the input letters $a,b$ that act as follows:
\[
i\dt a:=\begin{cases}
i+1 &\text{if } i=1,2,\\
9-i &\text{if } i=3,4,5,6;
\end{cases}
\qquad
i\dt b:=\begin{cases}
i+1 &\text{if } i<6,\\
1 &\text{if } i=6.
\end{cases}
\]
The automaton $\mathrsfs{E}'_6$ is shown in Fig.\,\ref{fig:eprime6}.
\begin{figure}[ht]
\begin{center}
\unitlength .5mm
\begin{picture}(72,72)(0,-86)
\gasset{Nw=16,Nh=16,Nmr=8} \node(n1)(56.0,-16.0){1} \node(n2)(68.0,-44.0){2}
\node(n3)(56.0,-72.0){3} \node(n4)(16.0,-72.0){4}
\node(n5)(4.0,-44.0){5} \node(n6)(16.0,-16.0){6}
\drawedge[ELdist=1.7](n1,n2){$a,b$} \drawedge[ELdist=1.7](n2,n3){$a,b$}
\drawedge[ELdist=1.7](n3,n4){$b$}
\drawedge[curvedepth=3,ELdist=1.7](n4,n5){$a,b$}
\drawedge[curvedepth=3,ELdist=1.7](n5,n4){$a$} 
\drawedge[ELdist=1.7](n5,n6){$b$}
\drawedge[ELdist=1.7](n6,n1){$b$}
\drawedge[curvedepth=3,ELdist=1.7](n6,n3){$a$}
\drawedge[curvedepth=3,ELdist=1.7](n3,n6){$a$}
\end{picture}
\end{center}
\caption{The automaton $\mathrsfs{E}'_6$}\label{fig:eprime6}
\end{figure}
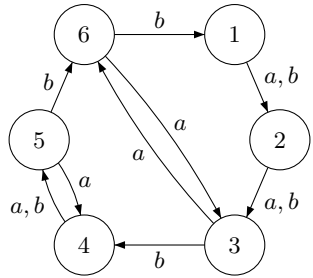
 
It can be verified that $\mathrsfs{E}'_6$ is not completely reachable, and moreover, it is not synchronizing. Let us run our algorithm on this DFA. As in the previous example, 
the first step of the algorithm produces 2 disjoint cycles $1\to 3\to 5\to 1$ and $2\to 4\to 6\to 2$. We then proceed by inspecting the words $a^2, aba, ab^2a, ab^3a, ab^4a, ab^5a$ whose actions are gathered in the following table:
\begin{center}
\begin{tabular}{c@{\rule{5pt}{0pt}}|@{\rule{5pt}{0pt}}c@{\rule{5pt}{0pt}}c@{\rule{5pt}{0pt}}c@{\rule{5pt}{0pt}}c@{\rule{5pt}{0pt}}c@{\rule{5pt}{0pt}}c}
        & 1 & 2 & 3 & 4 & 5 & 6\\
\hline
$a^2$\rule{0pt}{13pt}   & 3 & 6 & 3 & 4 & 5 & 6\\
$aba$\rule{0pt}{13pt}   & 6 & 5 & 2 & 3 & 4 & 5\\
$ab^2a$\rule{0pt}{13pt} & 5 & 4 & 3 & 2 & 3 & 4\\
$ab^3a$\rule{0pt}{13pt} & 4 & 3 & 6 & 3 & 2 & 3\\
$ab^4a$\rule{0pt}{13pt} & 3 & 2 & 5 & 6 & 3 & 2\\
$ab^5a$\rule{0pt}{13pt} & 2 & 3 & 4 & 5 & 6 & 3
\end{tabular}
\end{center}
Again, we see that $aba$ and $ab^5a$ are the only words of defect~1, and again, $ab^5a$ induces the same edge $1\to 3$ as $a$. We have $\dupl(aba)=5$ whence $aba$ adds the edge $1\to 5$; this new edge, however, fails to connect the cycle $1\to 3\to 5\to 1$ with $2\to 4\to 6\to 2$, and the algorithm stops.

\section{Minimal Completely Reachable Automata}
\label{sec:minimal}

\emph{Syntactic complexity of a regular language} is a well established concept that has attracted much attention lately, see, e.g., \cite{Brzozowski&Li:2014,Brzozowski&Szykula:2015}. It can be defined as the size of the transition monoid of the minimal DFA recognizing the language. It appears to be worthwhile to extend this concept to automata by defining the \emph{syntactic complexity} of an arbitrary DFA $\mathrsfs{A}$ as the size of its transition monoid $M(\mathrsfs{A})$. In fact, if one thinks of a DFA as a computational device rather than acceptor, its transition monoid can be thought of as the device's `software library' since the monoid contains exactly all programs (transformations) that the automaton can execute. From this viewpoint, measuring the complexity of an automaton by the size of its `software library' is fairly natural.

As already mentioned in Section~\ref{sec:intro}, the syntactic complexity of a \cran\ with $n$ states cannot be less than $2^n-1$. It turns out that this lower bound is tight if one considers automata over unrestricted alphabet. We present now a construction for \cra\ with $n$ states and syntactic complexity $2^n-1$; for short, we call them \emph{minimal \cra}.

Our construction produces minimal \cra\ from full binary trees satisfying certain subordination conditions. Recall that a binary tree is said to be \emph{full} if each its vertex $v$ either is a leaf or has exactly two children that we refer to as the \emph{left child} or the \emph{son} of $v$ and the \emph{right child} or the \emph{daughter} of~$v$. (Thus, all vertices except the root have a gender.) It is well known (and easy to verify) that a full binary tree with $n$ leaves has $2n-1$ vertices. As full binary trees are the only trees occurring in this paper, we call them just trees in the sequel.

If $\Gamma$ is a tree and $v$ is a vertex in $\Gamma$, we denote by $\Gamma_v$ the subtree of $\Gamma$ rooted at $v$. The \emph{span} of $v$, denoted $\spn(v)$, is the number of leaves in the subtree $\Gamma_v$. Fig.~\ref{fig:tree} shows a tree with vertices labelled by their spans.
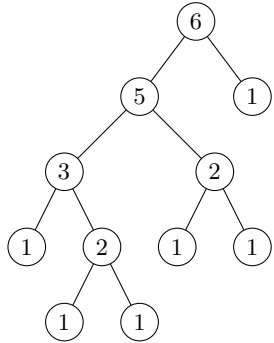
\begin{figure}[hb]
\begin{center}
\unitlength 0.50mm
\begin{picture}(80,78)(0,5)
\gasset{AHnb=0}
\gasset{Nw=10,Nh=10,Nmr=5}
\node(1)(10,30){$1$}
\node(2)(20,10){$1$}
\node(3)(40,10){$1$}
\node(4)(30,30){$2$}
  \drawedge(2,4){}
  \drawedge(3,4){}
\node(5)(20,50){$3$}
  \drawedge(1,5){}
  \drawedge(4,5){}
\node(6)(50,30){$1$}
\node(7)(70,30){$1$}
\node(8)(60,50){$2$}
  \drawedge(6,8){}
  \drawedge(7,8){}
\node(9)(40,70){$5$}
  \drawedge(5,9){}
  \drawedge(8,9){}
\node(10)(70,70){$1$}
\node(11)(55,90){$6$}
  \drawedge(9,11){}
  \drawedge(10,11){}
\end{picture}
\caption{An example of a tree with spans of its vertices shown}\label{fig:tree}
\end{center}
\end{figure}

By a \emph{homomorphism} between two trees $\Gamma_1$ and $\Gamma_2$ we mean a map from the vertex set of $\Gamma_1$ into the vertex set of $\Gamma_2$ that sends the root of $\Gamma_1$ to the root of $\Gamma_2$ and preserves the parent--child relation and the genders of non-root vertices. Given two trees $\Gamma_1$ and $\Gamma_2$, we say that $\Gamma_1$ \emph{subordinates} $\Gamma_2$ if there exists a 1-1 homomorphism $\Gamma_1\to\Gamma_2$. If $u$ and $v$ are two vertices of the same tree $\Gamma$, we say that $u$ \emph{subordinates} $v$ if the subtree $\Gamma_u$ subordinates the subtree $\Gamma_v$. A tree is said to be \emph{respectful} if it satisfies two
conditions:
\begin{enumerate}
\itemindent = .3cm
\item[(S1)] if a male vertex has a nephew, the nephew subordinates his uncle;
\item[(S2)] if a female vertex has a niece, the niece subordinates her aunt.
\end{enumerate}

For an illustration, the tree shown in Fig.~\ref{fig:tree} satisfies (S1) but fails to satisfy (S2): the daughter of the root has a niece but this niece does not subordinates her aunt. On the other hand, the tree shown in Fig.~\ref{fig:respectful tree} is respectful. (In order to ease the inspection of this claim, we have shown the uncle--nephew and the aunt--niece relations in this tree with dotted and dashed arrows respectively.)
\begin{figure}[hb]
\begin{center}
\unitlength 0.50mm
\begin{picture}(100,78)(0,5)
\gasset{AHnb=0}
\gasset{Nw=10,Nh=10,Nmr=5}
\node(1)(10,30){$1$}
\node(2)(20,10){$1$}
\node(3)(40,10){$1$}
\node(4)(30,30){$2$}
  \drawedge(2,4){}
  \drawedge(3,4){}
\node(5)(20,50){$3$}
  \drawedge(1,5){}
  \drawedge(4,5){}
\node(6)(50,30){$1$}
\node(7)(70,30){$1$}
\node(8)(60,50){$2$}
  \drawedge(6,8){}
  \drawedge(7,8){}
\node(9)(40,70){$5$}
  \drawedge(5,9){}
  \drawedge(8,9){}
\node(10)(90,70){$2$}
\node(11)(55,90){$7$}
  \drawedge(9,11){}
  \drawedge(10,11){}
\node(12)(80,50){$1$}
\node(13)(100,50){$1$}
  \drawedge(12,10){}
  \drawedge(13,10){}
\gasset{AHnb=1,dash={0.2 1.2}0}
  \drawedge[sxo=7,syo=-3,exo=-6,eyo=5](9,12){}
  \drawedge[sxo=7,syo=-3,exo=-7,eyo=4](5,6){}
  \drawedge[sxo=4,syo=-6,exo=-2,eyo=4](1,2){}
\gasset{AHnb=1,dash={1.5}{1.5}}
  \drawedge[curvedepth=-2,exo=3.5,eyo=5.5](10,8){}
  \drawedge[curvedepth=-2,exo=3.5,eyo=5.5](8,4){}
\end{picture}
\caption{An example of a respectful tree}\label{fig:respectful tree}
\end{center}
\end{figure}
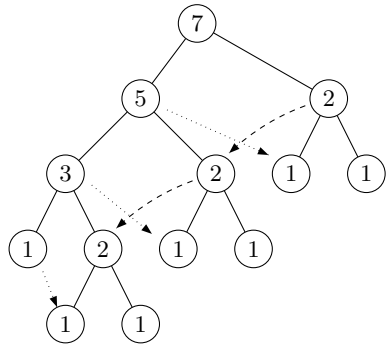

It is easy to show that there exist respectful trees with any number of leaves. In the following table (borrowed from~\cite{Bondar:2016}) we present the numbers of respectful trees with up to 10 leaves.
\begin{center}
\begin{tabular}[t]{|p{4.4cm}|c|c|c|c|c|c|c|c|c|c|c|c|c|c|c|}
\hline
Number of leaves\rule{0pt}{12pt}   &\rule{5pt}{0pt}1\rule{5pt}{0pt}&\rule{5pt}{0pt}2\rule{5pt}{0pt}&\rule{5pt}{0pt}3\rule{5pt}{0pt}
&\rule{5pt}{0pt}4\rule{5pt}{0pt}&\rule{5pt}{0pt}5\rule{5pt}{0pt}&\rule{5pt}{0pt}6\rule{5pt}{0pt}
&\rule{5pt}{0pt}7\rule{5pt}{0pt}&\rule{5pt}{0pt}8\rule{5pt}{0pt}&\rule{5pt}{0pt}9\rule{5pt}{0pt}
&\rule{3pt}{0pt}10\rule{5pt}{0pt} \\
\hline
Number of respectful trees\rule{0pt}{12pt} &1&1&2&3&6&10&18&32&58&101 \\
\hline
\end{tabular}
\end{center}
We are not aware of any closed formula for the number of respectful trees with a given number of leaves.

In our construction, we use certain markings of trees by intervals of the set $\mathbb{N}$ of positive integers considered as a chain under the usual order:
\[
1<2<\dots<n<\dotsc.
\]
If $i,j\in\mathbb{N}$ and $i\le j$, the \emph{interval} $[i,j]$ is the set $\{k\in X_n\mid i\le k\le j\}$. We write $[i]$ instead of $[i,i]$. Now, a \emph{faithful interval marking} of a tree $\Gamma$ is a map $\mu$ from the vertex set of $\Gamma$ into the set of all intervals in $\mathbb{N}$ such that for each vertex $v$,
\begin{itemize}
\item the number of elements in the interval $v\mu$ is equal to $\spn(v)$;
\item if $v\mu=[i,j]$ and $s$ and $d$ are respectively the son and the daughter of $v$, then $s\mu=[i,k]$ and $d\mu=[k+1,j]$ for some $k$ such that $i\le k<j$.
\end{itemize}
It easy to see that every tree $\Gamma$ admits a faithful interval marking which is unique up to an additive translation: given any two markings $\mu,\mu'$ of $\Gamma$, there is an integer $m$ such that $v\mu=v\mu'+m$ for every vertex $v$.
Observe that if $\mu$ is a faithful interval marking of a tree $\Gamma$ and $v$ is a vertex of $\Gamma$, then the restriction of $\mu$ to the subtree $\Gamma_v$ is a faithful interval marking of the latter. Fig.~\ref{fig:marking} demonstrates a faithful interval marking of the tree from Fig.~\ref{fig:respectful tree}.
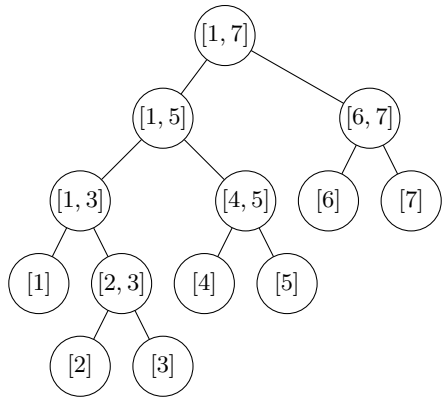
\begin{figure}[hb]
\begin{center}
\unitlength 0.55mm
\begin{picture}(100,81)(0,5)
\gasset{AHnb=0}
\node(1)(10,30){$[1]$}
\node(2)(20,10){$[2]$}
\node(3)(40,10){$[3]$}
\node(4)(30,30){$[2,3]$}
  \drawedge(2,4){}
  \drawedge(3,4){}
\node(5)(20,50){$[1,3]$}
  \drawedge(1,5){}
  \drawedge(4,5){}
\node(6)(50,30){$[4]$}
\node(7)(70,30){$[5]$}
\node(8)(60,50){$[4,5]$}
  \drawedge(6,8){}
  \drawedge(7,8){}
\node(9)(40,70){$[1,5]$}
  \drawedge(5,9){}
  \drawedge(8,9){}
\node(10)(90,70){$[6,7]$}
\node(11)(55,90){$[1,7]$}
  \drawedge(9,11){}
  \drawedge(10,11){}
\node(12)(80,50){$[6]$}
\node(13)(100,50){$[7]$}
  \drawedge(12,10){}
  \drawedge(13,10){}
\end{picture}
\caption{A faithful interval marking of the tree from Fig.~\ref{fig:respectful tree}}\label{fig:marking}
\end{center}
\end{figure}

We have prepared everything and can now present our construction.

\paragraph*{\textbf{Construction T2A (trees to automata)}.}
For each respectful tree $\Gamma$ with $n$ leaves and each its faithful interval marking $\mu$, we construct an automaton denoted by $\mathrsfs{A}_\mu(\Gamma)$. The states of $\mathrsfs{A}_\mu(\Gamma)$ are the elements of the interval $r\mu$, where $r$ stands for the root of $\Gamma$, and the input alphabet of $\mathrsfs{A}_\mu(\Gamma)$ consists of $2n-2$ letters $a_v$, one for each non-root vertex $v$ of $\Gamma$. To define the action of the letters, we proceed by induction on $n$. For $n=1$, that is, for the trivial tree $\Gamma$ with one vertex $r$ and no edges, $\mathrsfs{A}_\mu(\Gamma)$ is the trivial automaton with one state and no transitions, so that nothing has to be defined.

Now suppose that $n>1$. Take any non-root vertex $v$ of $\Gamma$; we have to define the action of the letter $a_v$ on the elements of the interval $r\mu$. If $s$ and $d$ are respectively the son and the daughter of $r$, the interval $r\mu$ is the disjoint union of $s\mu$ and $d\mu$. If $v\ne s$ and $v\ne d$, then $v$ is a non-root vertex in one of the subtrees $\Gamma_s$ or $\Gamma_d$. These two cases are symmetric, so that we may assume that $v$ belongs to $\Gamma_s$. By the induction assumption applied to $\Gamma_s$ and its marking induced by $\mu$, the action of $a_v$ is already defined on the states from the interval $s\mu$; we extend this action to the whole interval $r\mu$ by setting $y\dt a_v:=y$ for each $y\in d\mu$.

It remains to define the action of the letters $a_s$ and $a_d$. Again, by symmetry, it suffices to handle one of these cases, so that we define the action of $a_s$. First we define how $a_s$ acts on the interval $s\mu$. If $s$ has no nephew in $\Gamma$, then $d$ is a leaf and $d\mu=[y]$ for some $y\in\mathbb{N}$. Then we let $x\dt a_s:=y$ for each $x\in r\mu$. Otherwise let $t$ be the nephew of $s$. The subordination condition (S1) implies that there exists a 1-1 homomorphism $\xi\colon\Gamma_t\to\Gamma_s$. It is easy to see that the intervals $(\ell\xi)\mu$, where $\ell$ runs over the set of all leaves of the tree $\Gamma_t$, form a partition of the interval $s\mu$. Now we define the action of $a_s$ on $s\mu$ as follows: if a number $x\in s\mu$ belongs to $(\ell\xi)\mu$ for some leaf $\ell$ of $\Gamma_t$ and $\ell\mu=[y]$ for some $y\in\mathbb{N}$, we let $x\dt a_s:=y$.

By the induction assumption applied to the subtree $\Gamma_d$ and its marking induced by $\mu$, the action of the letter $a_t$ is already defined on the states from the interval $d\mu$; now we extend the action of $a_s$ to $d\mu$ by setting $y\dt a_s:=y\dt a_t$ for all $y\in d\mu$. This completes our construction.

\smallskip

The reader may find it instructive to work out Construction T2A on a concrete example. For the tree from Fig.~\ref{fig:respectful tree} and \ref{fig:marking} used for illustrations above, computing all 12 input letters of the corresponding automaton would be rather cumbersome but one can check, for instance, that the letters $a_s$ and $a_d$ act on the set $[1,7]$ as follows:
\[
a_s=\begin{pmatrix}
1&2&3&4&5&6&7\\
6&6&6&6&6&6&7
\end{pmatrix}
\quad
a_d=\begin{pmatrix}
1&2&3&4&5&6&7\\
1&1&1&2&3&4&5
\end{pmatrix}.
\]
Those who prefer a complete example can look at the DFA $\mathrsfs{E}_3$ from Example~\ref{examp:not necessary}:
the automaton was in fact derived by Construction T2A from the respectful tree with 3 leaves shown in Fig.~\ref{fig:t3}. In particular, this explains our choice of notation for the input letters of $\mathrsfs{E}_3$ that perhaps had slightly puzzled the reader when she or he encountered this automaton in Section~\ref{sec:sufficient}.
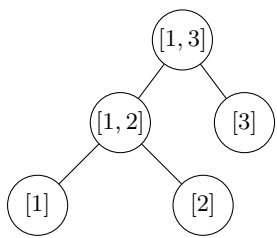
\begin{figure}[t]
\begin{center}
\unitlength 0.55mm
\begin{picture}(80,45)(0,35)
\gasset{AHnb=0}
\node(5)(20,50){$[1]$}
\node(8)(60,50){$[2]$}
\node(9)(40,70){$[1,2]$}
  \drawedge(5,9){}
  \drawedge(8,9){}
\node(10)(70,70){$[3]$}
\node(11)(55,90){$[1,3]$}
  \drawedge(9,11){}
  \drawedge(10,11){}
\end{picture}
\caption{The tree behind the automaton $\mathrsfs{E}_3$}\label{fig:t3}
\end{center}
\end{figure}
By the way, the flip-flop in Fig.~\ref{fig:flip-flop} also can be obtained by Construction~T2A (from the unique tree with 2 leaves).

Observe that all automata constructed from different markings of the same respectful tree are isomorphic since passing to another marking only results in a change of the state names. Taking this into account, we omit the reference to $\mu$ in the notation and denote the automaton derived from any marking of a given respectful tree $\Gamma$ simply by $\mathrsfs{A}(\Gamma)$.

We say that two DFAs $\mathrsfs{A}=\langle Q,\Sigma,\delta\rangle$ and $\mathrsfs{B}=\langle Q,\Delta,\zeta\rangle$ are \emph{syntactically equivalent} if their transition monoids coincide. Now we are ready for the main result of this section.

\begin{theorem}
\label{thm:minimal}
\textup1. For each respectful tree $\Gamma$, the automaton $\mathrsfs{A}(\Gamma)$ is a minimal \cran.

\textup2. Every minimal \cran\ is syntactically equivalent to an automaton of the form $\mathrsfs{A}(\Gamma)$ for a suitable respectful tree $\Gamma$.

\textup3. Every minimal \cran\ with $n$ states has at least $2n-2$ input letters.
\end{theorem}

Claims 1 and 2 in Theorem~\ref{thm:minimal} are essentially equivalent to the main results of the papers~\cite{Bondar:2014,Bondar:2016} by the first author who has used a slightly different construction expressed in the language of transformation monoids: given a marking of a  respectful tree $\Gamma$ she constructs the transition monoid of $\mathrsfs{A}(\Gamma)$ rather than the automaton itself. Claim~3 is new but we have not included its proof here due to the space limitations because the only proof we have at the moment requires reproducing several concepts and results from~\cite{Bondar:2014,Bondar:2016} and restating them in the language adopted in the present paper. It is very tempting to invent a direct proof of this claim that would bypass rather bulky considerations from~\cite{Bondar:2014,Bondar:2016}.

Theorem~\ref{thm:minimal} leaves widely open the question about lower bounds for syntactic complexity of \cra\ with restricted alphabet. In particular, the case of \cra\ with 2 input letters both is of interest and seems to be tractable. The latter conclusion follows from our analysis of \cra\ with 2 input letters at the end of Section~\ref{sec:complexity} which demonstrates that such DFAs have rather a specific structure.

We say that a DFA $\mathrsfs{A}=\langle Q,\Sigma,\delta\rangle$ \emph{induces} a DFA $\mathrsfs{B}=\langle Q,\Delta,\zeta\rangle$ on the same state set if the transition monoid of $\mathrsfs{A}$ contains that of  $\mathrsfs{B}$. Equivalently, this means that for every letter $b\in\Delta$, there exists a word $w\in\Sigma^*$ such that $\zeta(q,b)=\delta(q,w)$ for every $q\in Q$. This relation between automata plays an essential role in the theory of \sa, see, e.g., \cite{Ananichev&Gusev&Volkov:2013}. With respect to \cra, the following question is of interest: is it true that every \cran\ induces a minimal \cran? In other words, is it true that an automaton of the form  $\mathrsfs{A}(\Gamma)$ `hides' within every \cran?

\section{More Open Questions}
\label{sec:final}

Since \cra\ are synchronizing, it is natural to ask what is the maximum \rt\ for \cra\ with $n$ states. In view of Example~\ref{examp:cerny}, the lower bound $(n-1)^2$ for this maximum is provided by the \v{C}ern\'{y} automata $\mathrsfs{C}_n$. For \cra\ with 2 input letters this bound is tight because, except for the flip-flop, such automata have a letter that acts as a cyclic permutation of the state set, and therefore, Dubuc's result \cite{Du98} applies to them. Some partial results about synchronization of \cra\ can be found in~\cite{Don:2015}, but the general problem of finding the maximum \rt\ for \cra\ with $n$ states and unrestricted alphabet remains open.

The problem discussed in the previous paragraph basically asks what is the minimum length of a word that reaches a singleton. For \cra, a similar question makes sense for an arbitrary non-empty subset. Thus, we suggest to investigate the minimum length of a word that reaches a subset with $m$ element in a \cran\ with $n$ states as a function of $n$ and $m$. Don~\cite[Conjecture~2]{Don:2015} has formulated a very strong conjecture that implies the upper bound $n(n-m)$ on this length. Observe that if this upper bound indeed holds, then \cra\ satisfy the \v{C}ern\'{y} conjecture. To see this, take a \cran\ $\mathrsfs{A}=\langle Q,\Sigma\rangle$ with $n$ states; it should possess a letter $a\in\Sigma$ such that $q\dt a=q'\dt a$ for two different states $q,q'\in Q$. If a word $w\in\Sigma^*$ of length at most $n(n-2)$ is such that $Q\dt w=\{q,q'\}$, the word $wa$ is a \sw\ for $\mathrsfs{A}$ and has length at most $n(n-2)+1=(n-1)^2$.

Another intriguing problem about \cra\ suggested by the theory of \sa\ is a variant of the Road Coloring Problem.  We recall notions involved there. A \emph{road coloring} of a finite graph $\Gamma$ consists in assigning non-empty sets of labels (colors) from some alphabet $\Sigma$ to edges of $\Gamma$ such that the label sets assigned to the outgoing edges of each vertex form a partition of $\Sigma$. Colored this way, $\Gamma$ becomes a DFA over $\Sigma$; every such DFA is called a \emph{coloring} of $\Gamma$. Fig.\,\ref{fig:coloring} shows a graph and two of its colorings by $\Sigma=\{a,b\}$, one of which is the \v{C}ern\'{y} automaton $\mathrsfs{C}_4$.
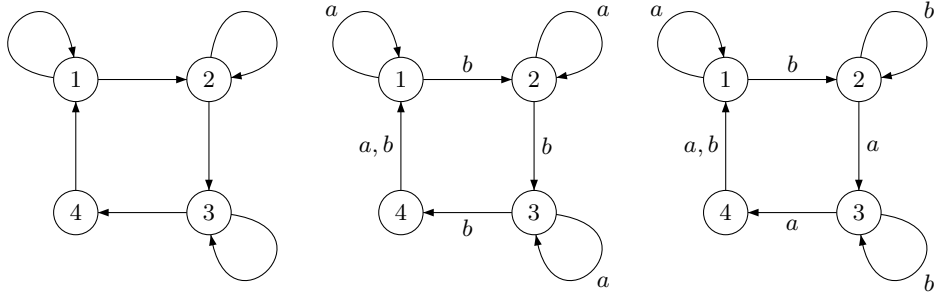
\begin{figure}[t]
 \begin{center}
  \unitlength=2.8pt
    \begin{picture}(18,30)(-63,-4)
    \gasset{Nw=6,Nh=6,Nmr=3}
    \node(A1)(0,18){$1$}
    \node(A2)(18,18){$2$}
    \node(A3)(18,0){$3$}
    \node(A4)(0,0){$4$}
    \drawloop[loopangle=135](A1){$a$}
    \drawloop[loopangle=45](A2){$b$}
    \drawloop[loopangle=-45](A3){$b$}
    \drawedge(A1,A2){$b$}
    \drawedge(A2,A3){$a$}
    \drawedge(A3,A4){$a$}
    \drawedge(A4,A1){$a,b$}
    \end{picture}
 \begin{picture}(18,30)(0,-4)
    \gasset{Nw=6,Nh=6,Nmr=3}
    \node(A1)(0,18){$1$}
    \node(A2)(18,18){$2$}
    \node(A3)(18,0){$3$}
    \node(A4)(0,0){$4$}
    \drawloop[loopangle=135](A1){$a$}
    \drawloop[loopangle=45](A2){$a$}
    \drawloop[loopangle=-45](A3){$a$}
    \drawedge(A1,A2){$b$}
    \drawedge(A2,A3){$b$}
    \drawedge(A3,A4){$b$}
    \drawedge(A4,A1){$a,b$}
    \end{picture}
 \begin{picture}(18,30)(63,-4)
    \gasset{Nw=6,Nh=6,Nmr=3}
    \node(A1)(0,18){$1$}
    \node(A2)(18,18){$2$}
    \node(A3)(18,0){$3$}
    \node(A4)(0,0){$4$}
    \drawloop[loopangle=135](A1){}
    \drawloop[loopangle=45](A2){}
    \drawloop[loopangle=-45](A3){}
    \drawedge(A1,A2){}
    \drawedge(A2,A3){}
    \drawedge(A3,A4){}
    \drawedge(A4,A1){}
    \end{picture}
 \end{center}
 \caption{A graph and two of its colorings}
 \label{fig:coloring}
\end{figure}
The Road Coloring Problem, recently solved by Trahtman~\cite{Tr09}, had asked which \scn\ graphs admit \emph{synchronizing colorings}, i.e.,  colorings that are \sa. It turns out that, as it was conjectured in~\cite{AGW}, the necessary and sufficient condition for a \scn\ graph to possess a synchronizing coloring is that the greatest common divisor of lengths of all directed cycles in the graph should be equal to~1. The latter property is called \emph{aperiodicity} or \emph{primitivity}.

\begin{figure}[hb]
\begin{center}
\unitlength=1.05mm
\begin{picture}(105,25)(-5,-5)
\node(A1)(0,0){1} \node(B1)(20,0){2} \node(C1)(10,17){3}
\drawedge[curvedepth=-3](A1,B1){}
\drawedge[curvedepth=-3](B1,C1){}
\drawedge[curvedepth=-3](C1,A1){}
\drawloop[loopangle=180](A1){}
\node(A2)(42,0){1} \node(B2)(59,0){2} \node(C2)(42,17){4} \node(D2)(59,17){3}
\drawedge[curvedepth=-3](A2,B2){}
\drawedge(B2,A2){}
\drawedge(C2,A2){}
\drawedge(B2,D2){}
\drawedge(D2,C2){}
\drawloop[loopangle=180](A2){}
\node(A3)(82,0){1} \node(B3)(99,0){2} \node(C3)(82,17){4} \node(D3)(99,17){3}
\drawedge[curvedepth=-3,ELside=r](A3,B3){$b,c$}
\drawedge[ELside=r](B3,A3){$c$}
\drawedge[ELside=r](C3,A3){$a,b,c$}
\drawedge[ELside=r](B3,D3){$a,b$}
\drawedge[ELside=r](D3,C3){$a,b,c$}
\drawloop[loopangle=180](A3){$a$}
\end{picture}
\caption{The left graph has no completely reachable coloring; the central graph has no completely reachable coloring with 2 letters but has a completely reachable coloring with 3 letters shown in the right}\label{fig:no cra coloring}
\end{center}
\end{figure}
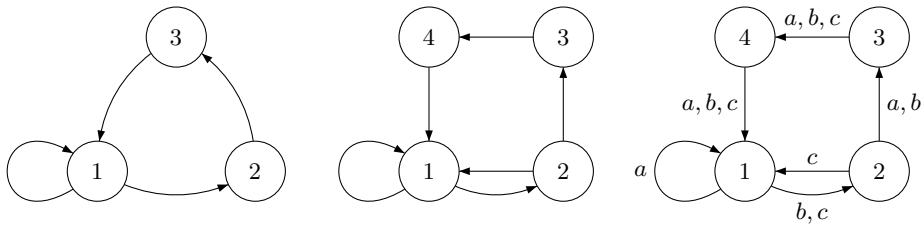
An analogous question makes sense for \cra. Namely, call a coloring of a graph \emph{completely reachable} if it yields a \cran. Our problem then consists in characterising graphs that admit completely reachable colorings. Such graphs must be \scn\ and primitive since every \cran\ is \scn\ and synchronizing. However, it is easy to produce an example of a \scn\ primitive graph that has no completely reachable coloring; such a graph is shown in Fig.~\ref{fig:no cra coloring} on the left. Moreover, there are interesting phenomena that have no parallel in the theory of \sa; for instance, there exist graphs that have no completely reachable coloring with 2 letters but admit such a coloring with 3 letters; an example of such a graph is presented in the center of Fig.~\ref{fig:no cra coloring} while the corresponding coloring is shown on the right.

\paragraph*{Acknowledgement.} The authors are grateful to Vladimir Gusev and Elena Pribavkina for a number of useful suggestions.

\end{document}